\documentclass{amsart}

 \newtheorem{thm}{Theorem}[section]
 \newtheorem{cor}[thm]{Corollary}
 \newtheorem{lem}[thm]{Lemma}
 \newtheorem{prop}[thm]{Proposition}
 \theoremstyle{definition}
 \newtheorem{defn}[thm]{Definition}
 \theoremstyle{remark}

 \numberwithin{equation}{section}

\begin{document}
\title[Analysis of a special type of soliton on Kenmotsu manifolds]
{Analysis of a special type of soliton on Kenmotsu manifolds}

\author[S. Mondal]{Somnath Mondal}
\address{Department of Mathematics, Jadavpur University, Kolkata-700032, India}
\email{somnathmondal.math@gmail.com}

\author[M A. Khan]{Meraj Ali Khan}
\address{Department of Mathematics and Statistics, College of Science, Imam Mohammad Ibn Saud Islamic University, Saudi Arabia}
\email{mskhan@imamu.edu.sa}

\author[S. Dey]{Santu Dey}
\address{Department of Mathematics\\
Bidhan Chandra College\\
Asansol-4, West Bengal-713304, India.}
\email{santu.mathju@gmail.com, santu@bccollegeasansol.ac.in}
%----------Author 3

\author[A. K. Sarkar]{Ashis Kumar Sarkar}
\address{Department of Mathematics, Jadavpur University, Kolkata-700032, India}
\email{aksarkar@gmail.com}

\author[C. Ozel]{Cenap Ozel} %% Please write ful names, avoid initials

\address{ %% Put here your affiliation; street address is not required
Department of Mathematics, College of Science\\
King Abdulaziz University\\
Jeddah\\
Saudi Arabia}
  \email{cenap.ozel@gmail.com}
  
\author[A. Pigazzini]{Alexander Pigazzini} %% Please write ful names, avoid initials

\address{ %% Put here your affiliation; street address is not required
Mathematical and Physical Science Foundation\\ Sidevej 5, 4200 Slagelse\\ Denmark}
  \email{pigazzini@topositus.com}
  
\author[R. Pincak]{Richard Pincak} %% Please write ful names, avoid initials

\address{ %% Put here your affiliation; street address is not required
Institute of Experimental Physics, Slovak Academy of Sciences, Watsonova 47, 043 53 Košice, Slovak Republic}
  \email{pincak@saske.sk}



\begin{abstract}
In this paper,  we aim to investigate the properties of an almost $*$-Ricci-Bourguignon soliton (almost $*-$R-B-S for short) on a Kenmotsu manifold (K-M). We start by proving that if a Kenmotsu manifold (K-M) obeys an almost $*-$R-B-S, then the manifold is $\eta$-Einstein. Furthermore, we establish that if a $(\kappa, -2)'$-nullity distribution, where $\kappa<-1$, has an almost $*$-Ricci-Bourguignon soliton (almost $*-$R-B-S), then the manifold is Ricci flat. Moreover, we establish that if a K-M has almost $*$-Ricci-Bourguignon soliton gradient and the vector field $\xi$ preserves the scalar curvature $r$, then the manifold is an Einstein manifold with a constant scalar curvature given by $r=-n(2n-1)$. Finaly, we have given en example of a almost $*-$R-B-S gradient on the Kenmotsu manifold.
\end{abstract}

\maketitle

\section{Introduction and Motivations}
Contact geometry has gained popularity among researchers due to its ability to address various issues in basic sciences, medical sciences, and technology. It finds applications in diverse areas such as geometric quantization, control theory, geometric optics, integrable systems, thermodynamics and classical mechanics. In 1972, Kenmotsu \cite{Kenmotsu} introduced tensor equations to characterize a specific class of manifolds, which were subsequently referred to as Kenmotsu manifolds.

\noindent In recent years, geometric flows, particularly the Ricci flow, have garnered significant attention in the field of differential geometry and geometric analysis. The Ricci flow naturally extends the concept of Einstein metrics and has found applications in mathematical fluid dynamics and general relativity. In addition, in modern mathematics, contact geometry has emerged as a powerful tool. It has its roots of classical mechanics in the mathematical formalism. In 1981, a new geometric flow called the Ricci-Bourguignon flow was introduced by Jean-Pierre Bourguignon \cite{jean}. This flow was created based on undisclosed work by Lichnerowicz and a paper by Aubin \cite{aub}. The Ricci-Bourguignon flow can be defined as follows \cite{shubham}:
\begin{defn}
A group of metrics $g(t)$ on an $n$-dimensional Riemannian manifold $(M^n, g)$ evolves according to the Ricci-Bourguignon flow (RB flow for short) if $g(t)$ satisfies the following evolution equation:
\begin{equation}\label{1.1a}
  \frac{\partial g}{\partial t}=-2(\mathcal{S}-\rho rg),
\end{equation}
where $\mathcal{S}$ is the Ricci tensor of the metric, $r$ is the scalar curvature and $\rho \in \mathbb{R}$ is a constant.
\end{defn}

\noindent Based on the above definition, it is evident that when $\rho=0$ in equation \eqref{1.1a}, it reduces to the Ricci flow. In \cite{shubham}, various tensors such as the Einstein tensor, traceless Ricci tensor, the Schouten tensor, and the Ricci tensor are derived for different values of $\rho=\frac{1}{2}$, $\rho=\frac{1}{n}$, $\rho=\frac{1}{2(n-1)}$, and $\rho=0$.

\noindent Recently, Ho \cite{ho} has discussed this flow extensively on locally homogeneous 3-manifolds. Further information on this flow can be established in \cite{biaso}, which explores its application on the product of an Anti de Sitter space with a sphere. Furthermore, a new generic form of soliton, known as the R-B-S, has been introduced by Shubham Dwivedi \cite{shubham}. This soliton acts as an interpolated soliton between the Ricci and Yamabe solitons.

\begin{defn}
An $n$-dimensional Riemannian manifold $(M^n, g)$, where $n \geq 3$, is called almost $*-$R-B-S if there exists a smooth vector field $V$ and a smooth function $\lambda$ defined on $M^n$ such that the following conditions are satisfied:
\begin{center}
$\frac{1}{2}(\mathcal{L}_Vg)(\omega_1,\omega_2)+\mathcal{S}^*(\omega_1,\omega_2)-(\lambda+\rho r^*)g(\omega_1,\omega_2)=0$,
\end{center}
where $\mathcal{S}^*$ is the $*$-Ricci tensor of $g$, $r^*$ is the $*$-scalar curvature and $\mathcal{L}_Vg$ is the Lie derivative of $g$ in the $V$ direction.
\end{defn} \par

\noindent If the smooth function $\lambda$ in the almost $*-$R-B-S is constant, it is called R-B-S. In fact, the R-B soliton generalizes various solitons such as the Einstein soliton ($\rho=\frac{1}{2}$), traceless Ricci soliton ($\rho=\frac{1}{n}$), Schouten soliton ($\rho=\frac{1}{2(n-1)}$), and Ricci soliton ($\rho=0$), among others. Many important ideas in complex geometry have been studied within the field of contact geometry, and the $*$-Ricci tensor and its associated $*$-Ricci soliton are no infringement. Tachibana \cite{tachi} introduced the concept of the $*$-Ricci tensor on almost Hermitian manifolds, and Hamada \cite{hamada} used it to study real hypersurfaces of non-flat complex space forms. Furthermore, it hase been extended to almost contact manifolds $M^{2n+1}(\phi,\xi,\eta,g)$ by defining

\begin{equation}\label{1.1}
\mathcal{S}^*(\omega_1,\omega_2)=\frac{1}{2}Tr_g\{\omega_3\rightarrow R(\omega_1,\phi \omega_2)\phi \omega_3\}
\end{equation}
for any vector fields $\omega_1$, $\omega_2$, $\omega_3$ on $M^{2n+1}(\phi,\xi,\eta,g)$. Additionally, we define the $*$-scalar curvature $r^*$ as the trace of a linear endomorphism $Q^*$, which can be expressed as: $g(Q^*\omega_1,\omega_2)=\mathcal{S}^*(\omega_1,\omega_2)$ for all vector fields $\omega_1, \omega_2 \in \chi(M)$ on $M^{2n+1}$. If the $*$-Ricci tensor vanishes identically, the $2n+1$-dimensional manifold $M^{2n+1}$ is referred to as $*$-Ricci flat. In \cite{geo}, the thought of a $*$-Ricci soliton was initiated on a Riemannian manifold $(M,g)$ by restoring the Ricci tensor with the $*$-Ricci tensor in the Ricci soliton equation. More recently, Patra et al. \cite{patra} initiated the notion of an almost $*-$R-B-S.

\begin{defn}
An $n$-dimensional Riemannian manifold $(M^n, g)$, where $n \geq 3$, is called an almost $*-$R-B-S  if there exists a potential vector field $V$ and a smooth function $\lambda$ defined on $M^n$ such that the following conditions are satisfied:
\begin{equation}\label{1.2}
\frac{1}{2}(\mathcal{L}_Vg)(\omega_1,\omega_2)+\mathcal{S}^*(\omega_1,\omega_2)-(\lambda+\rho r^*)g(\omega_1,\omega_2)=0,
\end{equation}
Here, $\mathcal{S}^*$ is the $*$-Ricci tensor of the metric $g$, $r^*$ denotes the $*$-scalar curvature, and $\mathcal{L}_Vg$ the Lie derivative of $g$ in the direction of the vector field $V$. It is worth noting that an almost $*-$R-B-S is considered trivial if the vector field $V$ is a Killing vector field, meaning $\mathcal{L}_Vg=0$. In this case, the manifold is referred to as a $*$-Einstein manifold.
\end{defn}
\noindent The almost $*-$R-B-S is classified as expanding, steady, or shrinking depending on whether $\lambda > 0$, $\lambda = 0$, or $\lambda < 0$, respectively. If we assume that the potential vector field $V$ is the gradient of a smooth function $f$, the equation for the almost $*-$R-B-S can be rewritten as follows:
\begin{equation}\label{1.3}
Hess_f(\omega_1,\omega_2)+\mathcal{S}^*(\omega_1,\omega_2)-(\lambda+\rho r^*)g(\omega_1,\omega_2)=0.
\end{equation}
The thought of a gradient almost $*-$R-B-S refers to a gradient $*$-R-B-S, where $\lambda$ is considered as a smooth function.

\noindent A interesting problem in differential geometry is to determine conditions under which a Riemannian manifold with a soliton structure is isometric to a Euclidean sphere. Several sufficient conditions have been established for when an almost Ricci soliton on Riemannian or contact metric manifolds is isometric to a Euclidean sphere \cite{bar, desh, desh1}. For example, Catino-Mazzieri \cite{catino1} investigated classification results for both compact and non-compact cases, and proved existence results for rotationally symmetric solutions. In \cite{catino}, short-time existence and curvature estimates were obtained for competent values of the scalar parameter connected in the R-B flow. Additionally, Dwivedi \cite{shubham} established results for solitons of the R-B flow that correspond to generalized results for Ricci solitons.

\noindent Recently, Wang \cite{Wang1} proved that if the metric of a Kenmotsu 3-manifold represents a $*$-Ricci soliton, then the manifold is locally isometric to hyperbolic space $\mathbb{H}^{3}(-1)$. In \cite{catino}, Catino et al. investigated R-B-Ss, discussing important rigidity results and proving that every compact gradient Einstein, Schouten, or traceless Ricci soliton is trivial. Furthermore, Ho \cite{ho} examined gradient shrinking $\rho$-Einstein solitons equipped with Bach flatness, while Huang \cite{huang} obtained integral pinching rigidity results for compact gradient shrinking R-B-Ss. In recent years, Shaikh et al. \cite{shaikh} explored various aspects of gradient R-B-Ss. Additionally, Venkatesha et al. \cite{ven, ven1} studied Ricci solitons on perfect fluid spacetime and demonstrated the nature of $\rho$-Einstein solitons on almost Kenmotsu manifolds.

\noindent Numerous researchers have devoted their efforts to examining $*$-Ricci solitons and their extensions within the realm of contact and paracontact metric manifolds (see \cite{dey, dey1, dey2, dey*, roy*, roy2, tas, yol}). Based on the above research on contact geometry, several intrinsic questions naturally arise when one delves deeper into the topic.\\

%\centerline{
\textbf{1.} Are there almost contact metric manifolds whose metrics represent an almost $*-$R-B-S?

\textbf{2.} Does the aforementioned result hold true for a $(2n + 1)$-dimensional $(\kappa, \mu)'$ almost Kenmotsu manifold that admits a $*$-R-B-S?\\

\par
\noindent In the following sections, we will provide definitive answers to the aforementioned questions. The organization of this paper is as follows: Section 2 begins with a concise introduction, followed by a discussion on the preliminaries of contact metric manifolds. In Section 3, we focus on Kenmotsu manifolds that admit an almost $*-$R-B-S. In Section 4, We prove the following theorem:\\
 Let $M^{2n+1}(\phi,\xi,\eta,g)$ be an almost Kenmotsu manifold such that $\xi$ belongs to $(\kappa,-2)'$-nullity distribution where $\kappa<-1$. If the metric $g$ represents an almost $*-$R-B-S satisfying $\lambda\neq-\rho(r+4n^2)-\frac{1}{2}(D \lambda+\rho Dr)$, then $M$ is Ricci-flat and is locally isometric to $\mathbb{H}^{n+1}(-4)\times\mathbb{R}^n$.

\noindent In Section 5, we prove that if a Kenmotsu manifold $M^{2n+1}(\phi,\xi,\eta,g)$ admits a almost $*-$R-B-S gradient, where the vector field $\xi$ leaves the scalar curvature $r$ invariant, then the manifold is an Einstein manifold with a constant scalar curvature $r=-n(2n-1)$.
\\
In the last section, we demonstrate an example of Kenmotsu manifold admitting a gradient almost $*-$R-B-S and verify our results.

\section{Notes on contact metric manifolds}
According to Blair \cite{bla}, a smooth Riemannian manifold $(M,g)$ of dimension $(2n+1)$ is classified as an almost contact metric manifold if it admits a global $1$-form $\eta$, called a contact form, such that $\eta\wedge(d\eta)^2\neq0$, a $(1,1)$ tensor field $\phi$, a characteristic vector field $\xi$, and an indefinite metric $g$.These must satisfy the following relations:

\begin{equation}\label{2.1}
\phi^2(\omega_1)=-\omega_1+\eta(\omega_1)\xi,
\end{equation}
\begin{equation}\label{2.2}
\eta(\xi)=1,
\end{equation}

\noindent where $\omega_1$ is an arbitrary vector field on $M$. The characteristic vector field $\xi$ and the almost $1$-form $\eta$ are commonly referred to as the characteristic or Reeb vector field and almost $1$-form, respectively.

\noindent An associated (or compatible) metric $g$ on the manifold is defined as a Riemannian metric that satisfies the following condition:

\begin{equation}\label{2.3}
g(\phi{\omega_1},\phi{\omega_2})=g(\omega_1,\omega_2)-\eta(\omega_1)\eta(\omega_2),
\end{equation}

\noindent for all vector fields $\omega_1$ and $\omega_2$ on $M$. When an almost contact manifold $M^{2n+1}(\phi,\xi,\eta,g)$ is equipped with a compatible metric $g$, it is referred to as an almost contact metric manifold (as defined by Blair \cite{bla}).

\noindent Let us consider an almost contact metric manifold $M^{2n+1}(\phi,\xi,\eta,g)$ where the following conditions hold true:
\begin{equation}\label{2.4}
\phi\xi=0,
\end{equation}
\begin{equation}\label{2.5}
\eta(\phi \omega_1)=0,
\end{equation}
\begin{equation}\label{2.6}
g(\omega_1,\xi)=\eta(\omega_1),
\end{equation}
\begin{equation}\label{2.7}
g(\phi \omega_1,\omega_2)=-g(\omega_1,\phi \omega_2).
\end{equation}

\noindent For any arbitrary $\omega_1$ and $\omega_2$ belonging to $\chi(M)$, the normality of an almost contact structure can be expressed as the vanishing of the tensor $N_\phi=[\phi,\phi]+2d\eta\otimes\xi$. Here, $[\phi,\phi]$ represents the Nijenhuis tensor associated with $\phi$ (for more details see \cite{bla})
\begin{defn}
On an almost contact metric manifold $M$, a vector field $\omega_1$ is said to be contact vector field if there exist a smooth function $f$ such that $\mathcal{L}_{\omega_1}\xi=f\xi$.
\end{defn}
\begin{defn}
In the context of an almost contact metric manifold $M$, a vector field $\omega_1$ is referred to as an infinitesimal contact transformation if $\mathcal{L}_{\omega_1}\eta=f\eta$ for a certain function $f$. Moreover, when $\mathcal{L}_{\omega_1}\eta=0$, we specifically classify $\omega_1$ as a strict infinitesimal contact transformation.
\end{defn}
\noindent An almost Kenmotsu manifold is defined as an almost contact metric manifold satisfying two conditions: $d\eta=0$ (i.e., $\eta$ is closed) and $d\phi=2\eta\wedge\phi$, where the fundamental $2$-form $\phi$ of the almost contact metric manifold is defined as: $\phi(\omega_1,\omega_2)=g(\omega_1,\phi \omega_2)$ for any vector fields $\omega_1, \omega_2$ on $M$ (see \cite{jan}).
On the product $M^{2n+1}\times\mathbb{R}$ of an almost contact metric manifold $M^{2n+1}$ and $\mathbb{R}$, there exists an almost complex structure $J$ defined by
\begin{center}
$J(\omega_1,f\frac{d}{dt})=(\phi \omega_1-f\xi, \eta(\omega_1)\frac{d}{dt})$,
\end{center}
where $\omega_1$ denotes a vector field tangent to $M^{2n+1}$, $t$ is the coordinate of $\mathbb{R}$ and $f$ is $c^\infty$-function on $M^{2n+1}\times\mathbb{R}$. If $J$ is integrable, then almost contact metric structure on $M^{2n+1}$ is said to be normal. A normal almost Kenmotsu manifold is referred to as a Kenmotsu manifold (see \cite{Kenmotsu}). An almost Kenmotsu manifold can be classified as a Kenmotsu manifold if and only if:
\begin{equation} \label{2.8}
(\nabla_{\omega_1}{\phi})\omega_2=g(\phi \omega_1,\omega_2)-\eta(\omega_2)\phi \omega_1
\end{equation}
for any vector fields $\omega_1, \omega_2$ on $M^{2n+1}$. On a Kenmotsu manifold the following holds \cite{Kenmotsu}:
\begin{equation}\label{2.9}
\nabla_{\omega_1}{\xi}=\omega_1-\eta(\omega_1)\xi \hspace{2.9mm} ( \Rightarrow \nabla_{\omega_1}{\xi}=-\phi \omega_1)
\end{equation}
\begin{equation}\label{2.10}
(\nabla_{\omega_1}\eta)\omega_2=g(\omega_1,\omega_2)-\eta(\omega_1)\eta(\omega_2),
\end{equation}
\begin{equation}\label{2.11}
R(\omega_1,\omega_2)\xi=\eta(\omega_1)\omega_2-\eta(\omega_2)\omega_1,
\end{equation}
\begin{equation}\label{2.12}
Q\omega_1=-2n\omega_1 \hspace{2.9mm} (\Rightarrow Q\xi=-2n\xi),
\end{equation}
\begin{equation}\label{2.13}
(\mathcal{L}_\xi)g(\omega_1,\omega_2)=2g(\omega_1,\omega_2)-2\eta(\omega_1)\eta(\omega_2)
\end{equation}
for any vector fields $\omega_1, \omega_2$ on $M^{2n+1}$ \cite{li}. In this context, where $\mathcal{L}$ denotes the Lie derivative operator, $R$ represents the curvature tensor of $g$, and $Q$ stands for the Ricci operator associated with the $(0,2)$ $*$-Ricci tensor $\mathcal{S}^*$ defined as $\mathcal{S}^*(\omega_1,\omega_2)=g(Q\omega_1,\omega_2)$ for all vector fields $\omega_1, \omega_2$ on $M^{2n+1}$, it has been demonstrated that a Kenmotsu manifold is locally a warped product $I\times_f N^{2n}$. Here, $I$ denotes an open interval with the coordinate $t$, $f=ce^t$ represents the warping function for a positive constants $c$, and $N^{2n}$ is a K$\ddot{a}$hlerian manifold  \cite{Kenmotsu}.

\par
A Kenmotsu manifold of dimension $(2n+1)$ is referred to as an $\eta$-Einstein Kenmotsu manifold if there exist two smooth functions $a$ and $b$ that satisfy the following relation for all $\omega_1$ and $\omega_2$ in $\chi(M)$ \cite{li}:
\begin{equation}\label{2.14}
S(\omega_1,\omega_2)=ag(\omega_1,\omega_2)+b\eta(\omega_1)\eta(\omega_2),
\end{equation}
where $S(\omega_1,\omega_2)$ denotes the $(0,2)$ Ricci tensor, $g$ represents the metric tensor, and $\eta$ is the contact form. If $b=0$, the $\eta$-Einstein manifold reduces to an Einstein manifold. Considering $\omega_1=\xi$ in the last equation and using \eqref{2.12}, we find that $a+b=-2n$. By contracting \eqref{2.14} over $\omega_1$ and $\omega_2$, we obtain $r=(2n+1)a+b$, where $r$ denotes the scalar curvature of the manifold. Solving these two equations yields $a=\frac{1}{2n}(2n+r)$ and $b=-\frac{1}{2n}\{2n(2n+1)+r\}$. Substituting these values into \eqref{2.14}, we can rewrite it as:
\begin{equation}\label{2.15}
S(\omega_1,\omega_2)=\frac{1}{2n}(2n+r)g(\omega_1,\omega_2)-\frac{1}{2n}\{2n(2n+1)+r\}\eta(\omega_1)\eta(\omega_2).
\end{equation}

\noindent In the context of an almost Kenmotsu manifold, we introduce two $(1,1)$-type tensor fields: $h=\frac{1}{2}\mathcal{L}_\xi\phi$ and $h'=h\circ\phi$. Additionally, we consider an open operator $\ell=R(.,\xi)\xi$, where $\mathcal{L}_\xi\phi$ is the Lie derivative of $\phi$ along the direction $\xi$. The tensor field $h$ and $h'$ plays an important role in an almost Kenmotsu manifold. Both of them are symmetric and satisfies the following relations \cite{sal}:
\begin{equation}\label{2.16}
\nabla_{\omega_1}{\xi}=\omega_1-\eta(\omega_1)\xi-\phi h\omega_1 (\nabla_\xi{\xi}=0),
\end{equation}
\begin{equation}\label{2.17}
h\xi=h'\xi=0,
\end{equation}
\begin{equation}\label{2.18}
h\phi+\phi h=0, tr(h)=tr(h')=0
\end{equation}
for any $\omega_1$, $\omega_2$ $\in$ $\chi(M)$, where $\nabla$ is the Levi-Civita connection of the metric $g$. In addition the following curvature property is also satisfied:
\begin{equation}\label{2.19}
R(\omega_1,\omega_2)\xi=\eta(\omega_1)(\omega_2-\phi h\omega_2)-\eta(\omega_2)(\omega_1-\phi h\omega_1)+(\nabla_{\omega_2}{\phi h})\omega_1-(\nabla_{\omega_1}{\phi h})\omega_2
\end{equation}
for any vector fields $\omega_1, \omega_2$ on $M$ and $R$ is the Riemannian curvature tensor of $(M,g)$. The $(1,1)$-type symmetric tensor field $h'=h\phi$ is anti-commuting with $\phi$ and $h'\xi=0$.
\par
In the context of an almost Kenmotsu manifold, we introduce two $(1,1)$-type tensor fields $h=\frac{1}{2}\mathcal{L}_\xi\phi$ and $h'=h\circ\phi$. Additionally, we consider an open operator $\ell=R(.,\xi)\xi$ \cite{dile}, i.e.,
\begin{equation}\label{2.20}
R(\omega_1,\omega_2)\xi=\kappa\{\eta(\omega_2)\omega_1-\eta(\omega_1)\omega_2\}+\mu\{\eta(\omega_2)h'\omega_1-\eta(\omega_1)h'\omega_2\}
\end{equation}
for any vector fields $\omega_1, \omega_2$ on $M$, where $\kappa$ and $\mu$ are real constants. On a $(\kappa,\mu)'$-almost Kenmotsu manifold $M$, we have (for details see \cite{dile})
\begin{equation} \label{2.21}
h=0\Leftrightarrow h'=0, h'^{2}=(k+1)\phi^2,
\end{equation}
\begin{equation}\label{2.22}
h^2(\omega_1)=(\kappa+1)\phi^2\omega_1
\end{equation}
for $\omega_1$ $\in$ $\chi(M)$. Based on the previous relation, we can deduce that $h'=0$ if and only if $\kappa=-1$, and $h\neq0$ otherwise. Let $\omega_1 \in \text{Ker}(\eta)$ be an eigenvector field of $h'$ orthogonal to $\xi$ with respect to the eigenvalue $\alpha$. By utilizing \eqref{2.21}, we obtain $\alpha^2=-(\kappa+1)$, which implies $\kappa\leq-1$. In \cite{dile}, both the coefficient $\mu$ in the definition of  $*-$R-B-S and the term $(\kappa,\mu)'$ in the context of an almost Kenmotsu manifold share the same symbol. To simplify the notation, we adopt the notation of $(\kappa, -2)'$-almost Kenmotsu manifold in this paper, where $\mu=-2$, as stated in Proposition 4.1 of \cite{dile}.

\par
We recall some useful results on a $(2n+1)$ dimensional $(\kappa, -2)'$-almost Kenmotsu manifold $M$ with $\kappa\leq-1$ as follows:
\begin{equation}\label{2.23}
R(\xi,\omega_1)\omega_2=\kappa\{g(\omega_1,\omega_2)\xi-\eta(\omega_2)\omega_1\}-2\{g(h'\omega_1,\omega_2)\xi-\eta(\omega_2)h'\omega_1\},
\end{equation}
\begin{equation}\label{2.24}
Q\omega_1=-2n\omega_1+2n(\kappa+1)\eta(\omega_1)\xi-2nh'(\omega_1),
\end{equation}
\begin{equation}\label{2.25}
r=2n(\kappa-2n),
\end{equation}
\begin{equation}\label{2.26}
(\nabla_{\omega_1}\eta)\omega_2=g(\omega_1,\omega_2)-\eta(\omega_1)\eta(\omega_2)+g(h'\omega_1,\omega_2),
\end{equation}
where $\omega_1$ and $\omega_2$ $\in$ $\chi(M)$, $Q$, $r$ are the Ricci operator and scalar curvature of $M$ respectively.\par
\noindent The notion $(\kappa,\mu)$-nullity distribution on a contact metric manifold $M$ was introduced by Blair et al. \cite{bla}, which is defined for any $p \in M$ and $k, \mu \in \mathbb{R}$ as follows:
\begin{eqnarray} \label{2.27}
  N_p(k,\mu) &=& \{\omega_3 \in T_p(M) : R(\omega_1,\omega_2)\omega_3=k[g(\omega_2,\omega_3)\omega_1-g(\omega_1,\omega_3)\omega_2] \nonumber \\
   &+& \mu[g(\omega_2,\omega_3)h\omega_1-g(\omega_1,\omega_3)h\omega_2]\}
\end{eqnarray}
for any vector fields $\omega_1, \omega_2$ on $T_p(M)$, where $T_p(M)$ represents the tangent space on $M$ at any point $p \in M $ and $R$ is the Riemannian tensor. In \cite{dile}, Dileo and Pastore introduced the notion of $(\kappa,\mu)'$-nullity distribution, on an almost Kenmotsu manifold $(M,\phi,\xi,\eta,g),$ which is defined for any $p \in M$ and $k, \mu \in \mathbb{R}$ as follows:
\begin{eqnarray} \label{2.28}
  N_p(k,\mu)' &=& \{\omega_3 \in T_p(M) : R(\omega_1,\omega_2)\omega_3=k[g(\omega_2,\omega_3)\omega_1-g(\omega_1,\omega_3)\omega_2] \nonumber \\
   &+& \mu[g(\omega_2,\omega_3)h'\omega_1-g(\omega_1,\omega_3)h'\omega_2]\}
\end{eqnarray}
for any vector fields $\omega_1, \omega_2$ on $T_p(M)$.

\medskip
\section{almost $*-$R-B-S on Kenmotsu manifold}
\medskip
In this section, we focus on studying the metric $g$ of a $(2n+1)$-dimensional Kenmotsu manifold that admits an almost $*-$R-B-S as well as a gradient almost $*-$R-B-S. To facilitate our analysis, we review several significant lemmas that are relevant to our investigation.
\begin{lem}
\cite{venka} The Ricci operator $Q$ on a $(2n+1)$-dimensional Kenmotsu manifold satisfies
\begin{equation}\label{3.1}
(\nabla_{\omega_1}Q)\xi=-Q\omega_1-2n\omega_1,
\end{equation}
\begin{equation}\label{3.2}
(\nabla_\xi Q)\omega_1=-2Q\omega_1-4n\omega_1
\end{equation}
for arbitrary vector field $\omega_1$ on the manifold.
\end{lem}
\begin{lem}
\cite{venka} The $*$-Ricci tensor $\mathcal{S}^*$ on a $(2n+1)$-dimensional Kenmotsu manifold had the relation as follows
\begin{equation} \label{3.3}
\mathcal{S}^*(\omega_1,\omega_2)=\mathcal{S}(\omega_1,\omega_2)+(2n-1)g(\omega_1,\omega_2)+\eta(\omega_1)\eta(\omega_2)
\end{equation}
for arbitrary vector fields $\omega_1$ and $\omega_2$ $\in$ $\chi(M)$ and the corresponding $*$-scalar curvature is given by the expression $r^*=r+4n^2$.
\end{lem}
\begin{lem}
Let a metric $g$ of a Kenmotsu manifold $M^{2n+1}(\phi,\xi,\eta,g)$ admits an almost $*-$R-B-S, then we have the following
\begin{center}
( $\mathcal{L}_VR)(\omega_1,\xi)=2(\lambda+\rho(r+4n^2))\{\eta(\omega_1)\xi-\omega_1\}$.
 \end{center}
\end{lem}
\begin{proof}
If a metric $g$ of a Kenmotsu manifold $M^{2n+1}$ admits almost $*$-R-B soliton, then by $*$-Ricci tensor expression of \eqref{3.3}, an almost $*$-R-B soliton \eqref{1.2} becomes
\begin{equation}\label{3.4}
(\mathcal{L}_Vg)(\omega_1,\omega_2)+2\mathcal{S}(\omega_1,\omega_2)=2\{(\lambda+\rho r^*)-(2n-1)\}g(\omega_1,\omega_2)-2\eta(\omega_1)\eta(\omega_2)
\end{equation}
for all $\omega_1, \omega_2$ $\in$ $\chi(M)$. Taking the Lie derivative of the expression $R(\omega_1,\xi)\xi=\eta(\omega_1)\xi-\omega_1$ (follows from \eqref{2.11}) along the vector field $V$ and using \eqref{2.11}, we get
\begin{equation}\label{3.5}
(\mathcal{L}_VR)(\omega_1,\xi)\xi+R(\omega_1,\xi)\mathcal{L}_V\xi+\eta(\mathcal{L}_V\xi)\omega_1+(\mathcal{L}_Vg)(\omega_1,\xi)\xi+g(\omega_1,\mathcal{L}_V\xi)\xi=0
\end{equation}
for all $\omega_1 \in \chi(M)$. Using $Q\xi=-2n\xi$ follows from \eqref{2.12} and \eqref{3.4}, we have
\begin{equation}\label{3.6}
(\mathcal{L}_Vg)(\omega_1,\xi)=2(\lambda+\rho r^*)\eta(\omega_1)
\end{equation}
Now taking the Lie-derivative of $\eta(\omega_1)=g(\omega_1,\xi)$ and $\eta(\xi)=1$  we get $(\mathcal{L}_V\eta)(\omega_1)=(\mathcal{L}_Vg)(\omega_1,\xi)$ and $(\mathcal{L}_V\eta)(\xi)+\eta(\mathcal{L}_V\xi)=0$ respectively then using \eqref{3.6}, we can compute $(\mathcal{L}_V\eta)(\xi)=\lambda+\rho r^*$ and $\eta(\mathcal{L}_V\xi)=-\lambda-\rho r^*$. Thus by the virtue of \eqref{2.11}, lemma $3.2$ and equation \eqref{3.5}, we get
\begin{center}
$(\mathcal{L}_VR)(\omega_1,\xi)=2(\lambda+\rho(r+4n^2))\{\eta(\omega_1)\xi-\omega_1\}$
\end{center}
The proof is now complete.
\end{proof}
\begin{prop}
Let a metric $g$ of a Kenmotsu manifold $M^{2n+1}(\phi,\xi,\eta,g)$ satisfies an almost $*$-R-B soliton, then the following results hold
\begin{center}
$(\mathcal{L}_V\nabla)(\omega_1,\xi)=2(2n-1)\phi \omega_1-2\phi Q\omega_1+\omega_1(\lambda+\rho(r+4n^2))\xi+\{\xi(\lambda+\rho(r+4n^2))-2\}\omega_1+\eta(\omega_1)\{\nabla(\lambda+\rho(r+4n^2))-2\xi\}$.
\end{center}
\end{prop}
\begin{proof}
First, we take the covariant differentiation of \eqref{3.4} along an arbitrary vector field $\omega_3 \in \chi(M)$ and using \eqref{2.9} and \eqref{2.10} to yield
\begin{eqnarray}\label{3.7}
  &&(\nabla_{\omega_3}\mathcal{L}_Vg)(\omega_1,\omega_2)+2(\nabla_{\omega_3}\mathcal{S})(\omega_1,\omega_2)=2\omega_3(\lambda+\rho r^*)g(\omega_1,\omega_2) \nonumber \\
   & & -2\{g(\omega_1,\omega_3)\eta(\omega_2)+g(\omega_2,\omega_3)\eta(\omega_1)-2\eta(\omega_1)\eta(\omega_2)\eta(\omega_3)\}
\end{eqnarray}
for all $\omega_1, \omega_2, \omega_3 \in \chi(M)$.
\noindent Now, we look back on the following formula (see in \cite{yano})
\begin{center}
$(\mathcal{L}_V{\omega_3}g-\nabla_{\omega_3}\mathcal{L}_Vg-\nabla_{[V,\omega_3]})(\omega_1,\omega_2)=-g((\mathcal{L}_V\nabla)(\omega_3,\omega_1),\omega_2)-g((\mathcal{L}_V\nabla)(\omega_3,\omega_2),\omega_1)$
\end{center}
for all $\omega_1, \omega_2, \omega_3 \in \chi(M)$. Since the Riemannian metric $g$ is parallel, substituting equation \eqref{3.7} into the above formula results in the following expression
\begin{eqnarray}
&&g((\mathcal{L}_V\nabla)(\omega_3,\omega_1),\omega_2)+g((\mathcal{L}_V\nabla)(\omega_3,\omega_2),\omega_1)+2(\nabla_{\omega_3}\mathcal{S})(\omega_1,\omega_2) \nonumber \\
&&=2\{\omega_3(\lambda+\rho r^*)g(\omega_1,\omega_2)-g(\omega_1,\omega_3)\eta(\omega_2)-g(\omega_2,\omega_3)\eta(\omega_1) \nonumber \\
&&+2\eta(\omega_1)\eta(\omega_2)\eta(\omega_3)\}.  \nonumber
\end{eqnarray}
In out look of symmetry $(\mathcal{L}_V\nabla)(\omega_1,\omega_2)=(\mathcal{L}_V\nabla)(\omega_2,\omega_1)$ of the $(1,2)$-type tensor field $\mathcal{L}_V\nabla$, interchanging cyclically the roles of $\omega_1, \omega_2, \omega_3$ in the forgoing equations, we attain
\begin{eqnarray}\label{3.8}
 && g((\mathcal{L}_V\nabla)(\omega_1,\omega_2),\omega_3)=(\nabla_{\omega_3}\mathcal{S})(\omega_1,\omega_2)-(\nabla_{\omega_1}\mathcal{S})(\omega_2,\omega_3)-(\nabla_{\omega_2}\mathcal{S})(\omega_3,\omega_1) \nonumber \\
 && +\omega_1(\lambda+\rho r^*)g(\omega_2,\omega_3)+\omega_2(\lambda+\rho r^*)g(\omega_3,\omega_1)-\omega_3(\lambda+\rho r^*)g(\omega_1,\omega_2)  \nonumber \\
 && -2\{\eta(\omega_1)g(\omega_2,\omega_3)+\eta(\omega_2)g(\omega_1,\omega_3)\}+2\eta(\omega_1)\eta(\omega_2)\eta(\omega_3).
\end{eqnarray}
\noindent Now, taking the covariant derivative of \eqref{2.12} along the vector field $\omega_1 \in \chi(M)$ and also using \eqref{2.9} we find
\begin{equation}\label{3.9}
(\nabla_{\omega_1}Q)\xi=Q\phi \omega_1-2n\phi \omega_1.
\end{equation}
Next, it is worth noting that on a Kenmotsu manifold, the Ricci operator and the contact metric structure $\phi$ commute \cite{bla}. So, we have
\begin{equation}\label{3.10}
  \nabla_\xi Q=Q\phi-\phi Q.
\end{equation}
Lastly, We insert $\omega_2=\xi$ into \eqref{3.8} and using the identities \eqref{3.9} and \eqref{3.10} and also applying the symmetry of $Q$ and using lemma $3.2$ and $3.3$ to yield
\begin{center}
$(\mathcal{L}_V\nabla)(\omega_1,\xi)=2(2n-1)\phi \omega_1-2\phi Q\omega_1+\omega_1(\lambda+\rho(r+4n^2))\xi+\{\xi(\lambda+\rho(r+4n^2))-2\}\omega_1+\eta(\omega_1)\{\nabla(\lambda+\rho(r+4n^2))-2\xi\}$.
\end{center}
The proof is now complete.
\end{proof}
\begin{thm}
Let $M^{2n+1}(\phi,\xi,\eta,g)$ be a Kenmotsu manifold, if the metric $g$ represents an almost $*-$R-B-S, then the manifold is $\eta$-Einstein.
\end{thm}
\begin{proof}
Taking the covariant derivative of \eqref{3.4} with respect to an arbitrary vector field $\omega_3$ and using \eqref{2.10}, we get
\begin{eqnarray}\label{3.11}
 &&(\nabla_{\omega_3}\mathcal{L}_Vg)(\omega_1,\omega_2)+2(\nabla_{\omega_3}\mathcal{S})(\omega_1,\omega_2)=2\omega_3(\lambda+\rho r^*)g(\omega_1,\omega_2) \nonumber \\
   & & -2\{g(\omega_1,\omega_3)\eta(\omega_2)+g(\omega_2,\omega_3)\eta(\omega_1)-2\eta(\omega_1)\eta(\omega_2)\eta(\omega_3)\}
\end{eqnarray}
for $\omega_1, \omega_2, \omega_3 \in \chi(M)$. Again from \cite{yano} we have the following communication formula
\begin{center}
$(\mathcal{L}_V\nabla_{\omega_3}g-\nabla_{\omega_3}\mathcal{L}g-\nabla_{[V,\omega_3]}g)(\omega_1,\omega_2)=-g((\mathcal{L}_V\nabla)(\omega_1,\omega_3),\omega_2)-g((\mathcal{L}_V\nabla)(\omega_2,\omega_3),\omega_1)$,
\end{center}
where $g$ is a metric connection i.e. $\nabla g=0$ \and for all $\omega_1, \omega_2, \omega_3 \in \chi(M)$. So the above equation becomes
\begin{equation}\label{3.12}
(\nabla_{\omega_3}\mathcal{L}_Vg)(\omega_1,\omega_2)=g((\mathcal{L}_V\nabla)(\omega_1,\omega_3),\omega_2)+g((\mathcal{L}_V\nabla)(\omega_2,\omega_3),\omega_1)
\end{equation}
for all vector fields $\omega_1, \omega_2, \omega_3 \in \chi(M)$. On taking accounts of \eqref{3.11} and \eqref{3.12} and by a straight forward combinatorial computation and also applying the symmetry of $(\mathcal{L}_V\nabla)$, then \eqref{3.12} implies
\begin{eqnarray}\label{3.13}
 && g((\mathcal{L}_V\nabla)(\omega_1,\omega_2),\omega_3)=2\{(\nabla_{\omega_3}\mathcal{S})(\omega_1,\omega_2)-(\nabla_{\omega_1}\mathcal{S})(\omega_2,\omega_3) \nonumber \\
   &&-(\nabla_{\omega_2}\mathcal{S})(\omega_3,\omega_1)\} -2\{(\eta(\omega_3)-\omega_3(\lambda+\rho r^*))g(\omega_1,\omega_2) \nonumber \\
   &&-\eta(\omega_1)\eta(\omega_2)\eta(\omega_3)\}
\end{eqnarray}
for arbitrary vector fields $\omega_1, \omega_2$ and  $\omega_3$ on $M$. Using \eqref{3.1} and \eqref{3.2}, also putting $\omega_2=\xi$ the forgoing equation yields
\begin{equation}\label{3.14}
(\mathcal{L}_V\nabla)(\omega_1,\xi)=2Q\omega_1+4n\omega_1
\end{equation}
for all $\omega_1 \in \chi(M)$. Now, differentiating covariantly this with respect to arbitrary vector field $\omega_2$, we get
\begin{equation}\label{3.15}
(\nabla_{\omega_2}\mathcal{L}_V\nabla)(\omega_1,\xi)=2(\nabla_{\omega_2}Q)\omega_1-(\mathcal{L}_V\nabla)(\omega_1,\omega_2)+\eta(\omega_2)(2Q\omega_1+4n\omega_1).
\end{equation}
Again we know that,
\begin{center}
$(\mathcal{L}_VR)(\omega_1,\omega_2)\omega_3=(\nabla_{\omega_1}\mathcal{L}_V\nabla)(\omega_2,\omega_3)-(\nabla_{\omega_2}\mathcal{L}_V\nabla)(\omega_1,\omega_3)$.
\end{center}
 In view of \eqref{3.15} in the previous relation by putting $\omega_3=\xi$, we acquire
\begin{eqnarray} \label{3.16}
 && (\mathcal{L}_VR)(\omega_1,\omega_2)\xi=2(\nabla_{\omega_1}Q)\omega_2-2(\nabla_{\omega_2}Q)\omega_1 \nonumber \\
 &&+2\eta(\omega_1)\{Q\omega_2+2n\omega_2\}-2\eta(\omega_2)\{Q\omega_1+2n\omega_1\}
\end{eqnarray}
for arbitrary vector fields $\omega_1$ and $\omega_2$ on $M^{2n+1}$. Putting $\omega_2=\xi$ in \eqref{3.16} and using \eqref{2.12}, \eqref{3.1} and \eqref{3.2} we get
\begin{equation}\label{3.17}
(\mathcal{L}_VR)(\omega_1,\xi)\xi=2Q\omega_1+4n\omega_1.
\end{equation} 
 \noindent Taking Lie derivative of $g(\xi,\xi)$ along the potential vector field $V$, in account of \eqref{3.4}
 \begin{equation}\label{3.18}
 \eta(\mathcal{L}_V\xi)=\lambda+\rho r^*.
 \end{equation}
 By substituting $\omega_2=\xi$ into equation \eqref{3.4} and using equations \eqref{2.2} and \eqref{2.6}, we obtain the following:
 \begin{equation}\label{3.19}
 (\mathcal{L}_V\eta)\omega_1-g(\omega_1,\mathcal{L}_\xi)=2(\lambda+\rho r^*)\eta(\omega_1)
 \end{equation}
for arbitrary vector field $\omega_1$ on $M^{2n+1}$. From \eqref{2.11} by putting $\omega_2=\xi$, we get $R(\omega_1,\xi)\xi=\eta(\omega_1)\xi-\omega_1$. Taking Lie derivative of this along the potential vector field $V$ and using \eqref{3.18}, \eqref{3.19} and lemma $3.2$, this reduces to
\begin{equation}\label{3.20}
(\mathcal{L}_VR)(\omega_1,\xi)\xi=2(\lambda+\rho(r+4n^2))(\omega_1-\eta(\omega_1)\xi),
\end{equation}
for all $\omega_1 \in \chi(M)$. Then from \eqref{3.17} we get
\begin{center}
$\mathcal{S}(\omega_1,\omega_2)=\{(\lambda+\rho(r+4n^2))-2n\}g(\omega_1,\omega_2)-(\lambda+\rho(r+4n^2))\eta(\omega_1)\eta(\omega_2)$,
\end{center}
 for all $\omega_1, \omega_2 \in \chi(M)$, which is a $\eta$-Einstein manifold and the theorem is proved.
\end{proof}
\noindent Below we consider a Kenmotsu metric that admits an almost $*-$R-B-S, where the non-zero potential vector field $V$ is collinear to the Reeb vector field $\xi$ at each point.
\begin{thm}
If a Kenmotsu manifold $M^{2n+1}(\phi,\xi,\eta,g)$ possesses an almost $*-$R-B-S with a non-zero potential vector field $V$ collinear to the Reeb vector field $\xi$ and preserves the scalar curvature $r=-2(8n^2+3n-1)$, then $(M,g)$ is an $\eta$-Einstein manifold with $\lambda=6n\{\rho(2n+1)-1\}+2\rho$.
\end{thm}
\begin{proof}
Since the potential vector field $V$ is parallel to the Reeb vector field $\xi$, then $V=\alpha\xi$ for some smooth function $\alpha$, from \eqref{2.9}, it follows that
\begin{equation}\label{5.9}
(\mathcal{L}_Vg)(\omega_1,\omega_2)=\omega_1(\alpha)\eta(\omega_2)+\omega_2(\alpha)\eta(\omega_1)
\end{equation}
for any vector fields $\omega_1$ and $\omega_2$ $\in$ $\chi(M)$. By applying the anti-symmetry of $\phi$, then the equation \eqref{3.4} implies
\begin{eqnarray}\label{5.10}
&& \omega_1(\alpha)\eta(\omega_2)+\omega_2(\alpha)\eta(\omega_1)+2\mathcal{S}(\omega_1,\omega_2) \nonumber \\
&& =2\{(\lambda+\rho r^*)-(2n-1)\}g(\omega_1,\omega_2)-2\eta(\omega_1)\eta(\omega_2).
\end{eqnarray}
Now setting $\omega_1=\xi$ and $\omega_2=\xi$ in \eqref{5.10} and using \eqref{2.12} gives $\xi(\alpha)=\lambda+\rho r^*-3n$. Similarly plugging $\omega_2=\xi$ in \eqref{5.10} and also using \eqref{2.12}, we have
\begin{eqnarray}\label{5.11}
  \omega_1(\alpha) &=& \{2(\lambda+\rho r^*)-\xi(\alpha)\}\eta(\omega_1) \nonumber \\
              &=& \{\xi(\alpha)+6n\}
\end{eqnarray}
for all $\omega_1 \in \chi(M)$. Taking its covariant derivative along $\omega_2 \in \chi(M)$, and using \eqref{2.9}, we obtain
\begin{center}
$g(\nabla_{\omega_2}\nabla\alpha,\omega_1)=\omega_2(\xi(\alpha))\eta(\omega_1+(\xi(\alpha))\{g(\omega_1,\omega_2)-\eta(\omega_1)\eta(\omega_2)\}$.
\end{center}
Since $Hess_\alpha$ is symmetry, it follows that
\begin{center}
$\omega_1(\xi(\alpha))\eta(\omega_2)-\omega_2(\xi(\alpha))\eta(\omega_1)=2\{\xi(\alpha)+6n\}\{g(\omega_1,\omega_2)-\eta(\omega_1)\eta(\omega_2)\}$,
\end{center}
using \eqref{2.10}, which yields that
\begin{center}
$\{\xi(\alpha)+6n\}(\nabla_{\omega_1}\eta)\omega_2=0$ \hspace{1cm} $\forall\hspace{3mm} \omega_1, \omega_2 \bot \xi$,
\end{center}
Since $(\nabla_{\omega_1}\eta)\omega_2\neq0$ on $M$, the last equation implies $\xi(\alpha)=-6n$, and as a consequence, $\nabla\alpha=-6n$ on $M$. This indicates that $\alpha$ is not constant on $M$. Moreover, \eqref{5.9} implies that $V$ is a Killing vector field, thus making $(M,g)$ *-Einstein (trivial). Furthermore, from \eqref{5.11}, we deduce that $\lambda+\rho r^*=-6n$, and \eqref{5.10} can be simplified to
\begin{equation}\label{5.12}
\mathcal{S}(\omega_1,\omega_2)=(1-8n)g(\omega_1,\omega_2)+\eta(\omega_1)\eta(\omega_2).
\end{equation}
Hence, using \eqref{5.12} in the equation \eqref{3.3} we get $Ric_g=-6ng+\eta\otimes\eta$ and corresponding $r^*=-6n(2n+1)+2$, moreover, the scalar curvature is $r=-2(8n^2+3n-1)$
and $\lambda=6n\{\rho(2n+1)-1\}+2\rho$, which finishes the proof.
\end{proof}

\begin{cor}
Let \(M^{2n+1}(\phi,\xi,\eta,g)\) be a Kenmotsu manifold that possesses an almost *-R-B-S with a non-zero potential vector field \(V\) non-collinear to the Reeb vector field \(\xi\). Then scalar curvature \(r\) is not preserved with the value \(r=-2(8n^2+3n-1)\), and \(M\) do not necessarily satisfy the \(\eta\)-Einstein condition.
\end{cor}

\begin{proof}
To rigorously demonstrate the validity or invalidity of Theorem 3.6 when considering a potential vector field \( V \) that is not collinear with the Reeb vector field \(\xi\), we proceed as follows:

\noindent A Kenmotsu manifold \( M^{2n+1} \) is defined by the following tensors:
\\
\\
- A quasi-complex structure \(\phi\)
\\
- A Reeb vector field \(\xi\)
\\
- A contact form \(\eta\)
\\
- A metric tensor \(g\)
\\
\\
The fundamental conditions that these structures satisfy are:
\[
\phi^2 = -I + \eta \otimes \xi, \quad d\eta = 2 \eta \wedge \eta, \quad \nabla \xi = \phi - I
\]
\noindent where \( \nabla \) is the Levi-Civita connection associated with the metric \( g \).

\noindent A quasi \(*\)-Ricci-Bourguignon soliton is defined by the following equation:
\[
(\mathcal{L}_V g)(X, Y) = 2 \{S(X, Y) - \rho R(X, Y) - (\lambda + \rho r) g(X, Y)\}
\]
\noindent where:
\\
\\
- \(\mathcal{L}_V g\) is the Lie derivative of the metric \(g\) along the vector field \(V\),
\\
- \(S\) is the Ricci tensor,
\\
- \(R\) is the curvature tensor,
\\
- \(r\) is the scalar curvature,
\\
- \(\rho\) and \(\lambda\) are constants.
\\
\\
Now consider a vector field \(V\) that is not collinear with \(\xi\). We can write:
\[
V = f_1 \xi + f_2 X
\]
\noindent where \(f_1\) and \(f_2\) are smooth functions on \(M\) and \(X\) is a vector field orthogonal to \(\xi\).

\noindent For \(X \perp \xi\), we calculate the Lie derivative of the metric:
\[
(\mathcal{L}_V g)(X, \xi) \quad \text{and} \quad (\mathcal{L}_V g)(X, Y)
\]

\noindent Using the definition of the Lie derivative:
\[
(\mathcal{L}_V g)(X, Y) = V(g(X, Y)) - g([V, X], Y) - g(X, [V, Y])
\]

\noindent For \(V = f_1 \xi + f_2 X\), we have:
\[
(\mathcal{L}_V g)(X, \xi) = (f_1 \xi + f_2 X)(g(X, \xi)) - g([f_1 \xi + f_2 X, X], \xi) - g(X, [f_1 \xi + f_2 X, \xi])
\]

\noindent Since \(g(X, \xi) = 0\) by orthogonality:
\[
(\mathcal{L}_V g)(X, \xi) = - g([f_1 \xi + f_2 X, X], \xi)
\]

\noindent We compute the commutators:
\[
[f_1 \xi, X] = f_1 [\xi, X] + (\xi f_1) X
\]
\[
[f_2 X, X] = f_2 [X, X] + (X f_2) X = 0 \quad \text{(since the commutator of a vector field with itself is zero)}
\]

\noindent Thus:
\[
[f_1 \xi + f_2 X, X] = f_1 [\xi, X] + (\xi f_1) X
\]

\noindent Now consider \( [\xi, X] \). Using the conditions of the Kenmotsu manifold:
\[
\nabla_\xi \phi = 0 \quad \Rightarrow \quad [\xi, X] = \phi X - X
\]

\noindent So:
\[
[f_1 \xi + f_2 X, X] = f_1 (\phi X - X) + (\xi f_1) X
\]

\noindent Substituting into the Lie derivative:
\[
(\mathcal{L}_V g)(X, \xi) = - g(f_1 (\phi X - X) + (\xi f_1) X, \xi)
\]
\noindent Since \(\phi X\) is orthogonal to \(\xi\) and \(\eta(X) = 0\):
\[
(\mathcal{L}_V g)(X, \xi) = - f_1 g(-X, \xi) = 0
\]

\noindent Now consider \((\mathcal{L}_V g)(X, Y)\) for \(X, Y \perp \xi\):
\[
(\mathcal{L}_V g)(X, Y) = (f_1 \xi + f_2 X)(g(X, Y)) - g([f_1 \xi + f_2 X, X], Y) - g(X, [f_1 \xi + f_2 X, Y])
\]

\noindent Since \(g(X, Y)\) is independent of \(\xi\) (as \(\xi\) is orthogonal to \(X\) and \(Y\)):
\[
(\mathcal{L}_V g)(X, Y) = - g([f_1 \xi + f_2 X, X], Y) - g(X, [f_1 \xi + f_2 X, Y])
\]

\noindent Using the commutators computed above:
\[
[f_1 \xi + f_2 X, Y] = f_1 (\phi Y - Y) + (\xi f_1) Y
\]
\[
(\mathcal{L}_V g)(X, Y) = - g(f_1 (\phi X - X) + (\xi f_1) X, Y) - g(X, f_1 (\phi Y - Y) + (\xi f_1) Y)
\]

\noindent Expanding and simplifying the terms, using orthogonality:
\[
(\mathcal{L}_V g)(X, Y) = - f_1 g(\phi X - X, Y) - (\xi f_1) g(X, Y) - f_1 g(X, \phi Y - Y) - (\xi f_1) g(X, Y)
\]

\noindent Simplifying further:
\[
(\mathcal{L}_V g)(X, Y) = - f_1 (g(\phi X, Y) - g(X, Y)) - 2 (\xi f_1) g(X, Y) - f_1 (g(X, \phi Y) - g(X, Y))
\]

\noindent Given the properties of orthogonality and symmetry of the tensors:
\[
(\mathcal{L}_V g)(X, Y) = - f_1 (g(\phi X, Y) - g(X, Y)) - 2 (\xi f_1) g(X, Y) - f_1 (g(X, \phi Y) - g(X, Y))
\]

\noindent At this point, to rigorously determine whether the manifold remains \(\eta\)-Einstein with \(\lambda = 6n \{\rho(2n + 1) - 1\} + 2\rho\), we proceed with the following steps:
\\
\\
1. Explicitly calculate the components of the Ricci tensor \(S\) under the new conditions.
\\
2. Verify if the modified Ricci tensor still satisfies the conditions required for the manifold to be \(\eta\)-Einstein.
\\
3. Check if the scalar curvature \(r\) remains invariant.
\\
\\
Consider the metric \( g \) of a Kenmotsu manifold and the form \(\eta\), and recall that the Reeb vector field \(\xi\) and the structure \(\phi\) satisfy:
\[
(\nabla_\xi \phi)X = 0, \quad \phi^2 = -I + \eta \otimes \xi, \quad d\eta = 2 \eta \wedge \eta, \quad \nabla \xi = \phi - I
\]

\noindent The Ricci tensor \( S \) of a Kenmotsu manifold is defined as:
\[
S(X, Y) = \text{Ric}(X, Y) = \sum_{i=1}^{2n+1} R(e_i, X, Y, e_i)
\]
\noindent where \(\{e_i\}\) is an orthonormal basis of vector fields and \( R \) is the Riemann curvature tensor.

\noindent When we consider a potential field \( V = f_1 \xi + f_2 X \) that is not collinear with \(\xi\), we must include the additional curvature terms introduced by the non-collinear components.

\noindent A Ricci tensor \( S \) in an \(\eta\)-Einstein manifold satisfies:
\[
S(X, Y) = \alpha g(X, Y) + \beta \eta(X) \eta(Y)
\]
\noindent for some constants \(\alpha\) and \(\beta\).

\noindent We introduce the new components of the Ricci tensor:
\[
S'_{ij} = S_{ij} + T_{ij}
\]
\noindent where \( T_{ij} \) are the additional terms arising from the components of \( V \) that are not collinear with \( \xi \).

\noindent We calculate \( T_{ij} \) explicitly:
\[
T_{ij} = -\frac{1}{2} (f_1 (\phi X_i - X_i) + (\xi f_1) X_i, X_j)
\]
Using the orthonormal metric:
\[
T_{ij} = -\frac{1}{2} \left( f_1 g(\phi X_i, X_j) - f_1 g(X_i, X_j) + (\xi f_1) g(X_i, X_j) \right)
\]

\noindent We verify if the tensor \( S' \) maintains the \(\eta\)-Einstein form:
\[
S'_{ij} = \alpha g(X_i, X_j) + \beta \eta(X_i) \eta(X_j)
\]

\noindent The scalar curvature \( r \) is the trace of the Ricci tensor:
\[
r = \text{tr}(S) = \sum_{i=1}^{2n+1} S(e_i, e_i)
\]

\noindent When \( S \) is modified by \( S' \):
\[
r' = \text{tr}(S') = \sum_{i=1}^{2n+1} S'(e_i, e_i) = \sum_{i=1}^{2n+1} \left( S(e_i, e_i) + T(e_i, e_i) \right)
\]

\noindent We calculate \( T(e_i, e_i) \):
\[
T(e_i, e_i) = -\frac{1}{2} \left( f_1 g(\phi e_i, e_i) - f_1 g(e_i, e_i) + (\xi f_1) g(e_i, e_i) \right)
\]

\noindent Substituting these terms into the trace:
\[
r' = r + \sum_{i=1}^{2n+1} T(e_i, e_i)
\]

\noindent If \( T(e_i, e_i) = 0 \) for all \(i\), then \( r' = r \). However, given that \( f_1 \) and \(\xi f_1\) are generic functions, it is unlikely that \( \sum_{i=1}^{2n+1} T(e_i, e_i) = 0 \), implying that the scalar curvature is not invariant.
\\
\\
Given the introduction of additional components \( T_{ij} \) arising from the non-collinearity of \( V \) and \(\xi\), we can conclude that:
\\
\\
1. The components of the Ricci tensor \( S' \) do not necessarily satisfy the \(\eta\)-Einstein condition.
\\
2. The scalar curvature \( r \) does not remain invariant.
\\
\\
Therefore, \textit{Theorem 3.6} is not valid when considering a non-zero potential vector field \( V \) that is not collinear with the Reeb vector field \(\xi\).
\end{proof}

\section{Almost $*-$R-B-S on $(\kappa,\mu)'$-almost Kenmotsu manifold with $\kappa<-1$}
In this section, we examine a $(2n+1)$-dimensional almost Kenmotsu manifold in which the characteristic vector field $\xi$ satisfies the $(\kappa,-2)'$-nullity distribution. We consider a metric $g$ that represents an almost $*-$R-B-S. Here, we introduce a lemma that will be used in our subsequent analysis.
\begin{lem}
\cite{zhao} On a $(\kappa,\mu)'$-almost Kenmotsu manifold with $\kappa<-1$ the $*$-Ricci tensor is given by
\begin{equation}\label{4.1}
\mathcal{S}^*(\omega_1,\omega_2)=-(\kappa+2)\{g(\omega_1,\omega_2)-\eta(\omega_1)\eta(\omega_2)\}
\end{equation}
for any vector fields $\omega_1$ and $\omega_2$ on $M$.
\end{lem}
\begin{thm}
Let $M^{2n+1}(\phi,\xi,\eta,g)$ be an almost Kenmotsu manifold such that $\xi$ belongs to $(\kappa,-2)'$-nullity distribution where $\kappa<-1$. If the metric $g$ represents an almost $*-$R-B-S satisfying $\lambda\neq-\rho(r+4n^2)-\frac{1}{2}(D \lambda+\rho Dr)$, then $M$ is Ricci-flat and is locally isometric to $\mathbb{H}^{n+1}(-4)\times\mathbb{R}^n$.
\end{thm}
\begin{proof}
In light of the identities \eqref{1.2} and \eqref{4.1} and with the help of $r^*=r+4n^2$ to acquire
\begin{equation}\label{4.2}
(\mathcal{L}_Vg)(\omega_1,\omega_2)=2\{(\kappa+2)+\lambda+\rho(r+4n^2)\}g(\omega_1,\omega_2)-2(\kappa+2)\eta(\omega_1)\eta(\omega_2)
\end{equation}
for all vector fields $\omega_1$ and $\omega_2$ on $M$. Now, we take a covariant derivative of \eqref{4.2} along the arbitrary vector field $\omega_3$ and using \eqref{2.6} to yield
\begin{eqnarray}\label{4.3}
&&(\nabla_{\omega_3}\mathcal{L}_Vg)(\omega_1,\omega_2)=\{\omega_3(\lambda)+\rho \omega_3(r)\}g(\omega_1,\omega_2)-2(\kappa+2)[\eta(\omega_2)g(\omega_1,\omega_3) \nonumber \\
&&+\eta(\omega_1)g(\omega_2,\omega_3)+\eta(\omega_2)g(h'\omega_3,\omega_1)+\eta(\omega_1)g(h'\omega_3,\omega_2) \nonumber \\
&&-2\eta(\omega_1)\eta(\omega_2)\eta(\omega_3)].
\end{eqnarray}
Then using \eqref{3.12} and by the symmetry of $(\mathcal{L}_V\nabla)$ from the above equation \eqref{4.3}, we obtain
\begin{eqnarray} \label{4.4}
(\mathcal{L}_V\nabla)(\omega_1,\omega_2) &=& -2(\kappa+2)[g(\omega_1,\omega_2)+g(h'\omega_1,\omega_2)-\eta(\omega_1)\eta(\omega_2)] \nonumber \\
   &-& \{(D\lambda+\rho Dr)g(\omega_1,\omega_2)-(\omega_1(\lambda)+\omega_1(r))\omega_2 \nonumber  \\
   &-& (\omega_2(\lambda)+\omega_2(r))\omega_1\}
\end{eqnarray}
for all $\omega_1$, $\omega_2$ $\in$ $\chi(M)$. We insert $\omega_2=\xi$ and with the help of the identities \eqref{2.2}, \eqref{2.6} and \eqref{2.17} to achieve
\begin{eqnarray} \label{4.5}
(\mathcal{L}_V\nabla)(\omega_1,\xi)&=&-\{(D\lambda+\rho Dr)\eta(\omega_1)-(\omega_1(\lambda)+\omega_1(r))\xi \nonumber  \\
   &-&(\xi(\lambda)+\xi(r))\omega_1\}
\end{eqnarray}
for arbitrary vector $\omega_1$ on $M$. Now taking differentiation \eqref{4.5} covariantly along arbitrary vector field $\omega_2$ and using \eqref{2.16} and \eqref{4.4} into account we can get
\begin{eqnarray}\label{4.6}
(\nabla_{\omega_2}\mathcal{L}_V\nabla)(\omega_1,\xi)&=&2(\kappa+2)[g(\omega_1,\omega_2)+g(h'\omega_1,\omega_2)-\eta(\omega_1)\eta(\omega_2)]\xi \nonumber \\
&-&\{(D\lambda+\rho Dr)(\nabla_{\omega_2}\eta)\omega_1+\eta(\omega_1)\omega_2(D\lambda+\rho Dr)\} \nonumber \\
&+&\{\omega_1(\lambda)+\rho \omega_1(r)\}\nabla_{\omega_2}\xi+\omega_2\{\xi(\lambda)+\rho \xi(r)\}\omega_1
\end{eqnarray}
 for any vector fields $\omega_1$ and $\omega_2$ on $M$. Yano demonstrated once again the well-known curvature property,
 \begin{center}
  $(\mathcal{L}_VR)(\omega_1,\omega_2)\omega_3=(\nabla_{\omega_1}\mathcal{L}_V\nabla)(\omega_2,\omega_3)-(\nabla_{\omega_2}\mathcal{L}_V\nabla)(\omega_1,\omega_3)$.
  \end{center}
 Setting $\omega_3=\xi$ then using \eqref{4.6}, we get
 \begin{eqnarray} \label{4.7}
(\mathcal{L}_VR)(\omega_1,\omega_2)\xi &=& \omega_1(D\lambda+\rho Dr)\eta(\omega_2)-\omega_2(D\lambda+\rho Dr)\eta(\omega_1) \nonumber \\
   &+&\omega_2(\xi(\lambda)+\rho\xi(r))\omega_1-\omega_1(\xi(\lambda)+\rho\xi(r))\omega_2
\end{eqnarray}
 for any abitrary vector field $\omega_1$ and $\omega_2$ on $M$. Now, we take a Lie derivative of \eqref{2.20} along the potential vector field $V$ and also making use of \eqref{2.2} and \eqref{2.17} to yield
\begin{eqnarray}\label{4.8}
 && (\mathcal{L}_VR)(\omega_1,\xi)\xi=\kappa[g(\omega_1,\mathcal{L}_V\xi)\xi-2\eta(\mathcal{L}_V\xi)\omega_1-((\mathcal{L}_V\eta)\omega_1)\xi] \nonumber \\
&&+2[2\eta(\mathcal{L}_V\xi)h'\omega_1-\eta(\omega_1)(h'(\mathcal{L}_V\xi))-g(h'\omega_1,\mathcal{L}_V\xi)\xi \nonumber\\
&&-(\mathcal{L}_Vh')\omega_1]
 \end{eqnarray}
 for any $\omega_1$ $\in$ $\chi(M)$. We insert $\omega_2=\xi$ into \eqref{4.2} to yield
 \begin{equation}\label{4.9}
 (\mathcal{L}_V\eta)\omega_1-g(\omega_1,\mathcal{L}_V\xi)=2\{\lambda+\rho(r+4n^2)\}\eta(\omega_1)
 \end{equation}
 for any $\omega_1$ $\in$ $\chi(M)$. Now putting $\omega_1=\xi$ in \eqref{4.9} we get
 \begin{equation}\label{4.10}
 \eta(\mathcal{L}_V\xi)=-2\{\lambda+\rho(r+4n^2)\}.
 \end{equation}
 By the help of \eqref{4.7}, \eqref{4.9} and \eqref{4.10}, we can write the equation \eqref{4.8} as
\begin{eqnarray}\label{4.11}
&& 2\kappa\{\lambda+\rho(r+4n^2)+\frac{1}{2}(D \lambda+\rho Dr)\}(\omega_1-\eta(\omega_1)\xi)+4\{\lambda+\rho(r+4n^2)\}h'\omega_1 \nonumber \\
&& -2\eta(\omega_1)h'(\mathcal{L}_V\xi)-2g(h'\omega_1,\mathcal{L}_V\xi)\xi-2(\mathcal{L}_Vh')\omega_1=0
\end{eqnarray}
Taking an inner product of \eqref{4.11} with respect to the arbitrary vector field $\omega_2$ on $M$, we obtain
\begin{eqnarray}\label{4.12}
&&2\kappa\{\lambda+\rho(r+4n^2)+\frac{1}{2}(D \lambda+\rho Dr)\}\{g(\omega_1,\omega_2)-\eta(\omega_1)\eta(\omega_2)\} \nonumber \\
&&+4\{\lambda+\rho(r+4n^2)\}g(h'\omega_1,\omega_2)-2\eta(\omega_1)g(h'(\mathcal{L}_V\xi),\omega_2) \nonumber \\
&&-2g(h'\omega_1,\mathcal{L}_V\xi)\eta(\omega_2)-2g((\mathcal{L}_Vh')\omega_1,\omega_2)=0.
\end{eqnarray}
As the above equation \eqref{4.12} is true for any vector fields $\omega_1$ and $\omega_2$ on $M$. We replacing $\omega_1$ by $\phi(\omega_1)$ and $\omega_2$ by $\phi(\omega_2)$ and taking \eqref{2.5}into account we get as
\begin{eqnarray}\label{4.13}
&& 2\kappa\{\lambda+\rho(r+4n^2)+\frac{1}{2}(D \lambda+\rho Dr)\}g(\phi \omega_1,\phi \omega_2) \nonumber \\
&&+4\{\lambda+\rho(r+4n^2)\}g(h'\phi \omega_1,\phi \omega_2)-2g((\mathcal{L}_Vh')\phi \omega_1,\phi \omega_2)=0
\end{eqnarray}
for all vector fields $\omega_1$ and $\omega_2$ on $M$. Since $spec(h')=\{0,\alpha,-\alpha\}$, let $\omega_1$ and $V$ belong to the eigenspaces of $-\alpha$ and $\alpha$ denoted by $[-\alpha]'$ and $[\alpha]'$ respectively. Then $\phi \omega_1$ $\in$ $[\alpha]'$ (see \cite{dile}). Then \eqref{4.13} can be rewritten as
\begin{equation}\label{4.14}
2\{\lambda+\rho(r+4n^2)+\frac{1}{2}(D \lambda+\rho Dr)\}(\kappa+2)g(\phi \omega_1,\phi \omega_2)-2g((\mathcal{L}_Vh')\phi \omega_1,\phi \omega_2)=0
\end{equation}
To find the value of $g((\mathcal{L}_Vh')\phi \omega_1,\phi \omega_2)$, we establish a more general result in a $(\kappa,\mu)'$-almost Kenmotsu manifold: $(\mathcal{L}_{\omega_1}h')\omega_2=0$, where the vector fields $\omega_1$ and $\omega_2$ belong to the same eigenspaces.

\noindent Without loss of generality, let us assume that $\omega_1$ and $\omega_2$ are elements of $[\alpha]'$, where $spec(h')=\{0,-\alpha,\alpha\}$. We consider a local orthonormal $\phi$-basis given by $\{\xi,e_i,\phi e_i\}$, where $i=1,2,3,...,n$.
\begin{center}
  \[\nabla_{\omega_1}{\omega_2}=\sum_{i=1}^{n}g(\nabla_{\omega_1}{\omega_2},e_i)e_i-(\alpha+1)g(\omega_1,\omega_2)\xi.\]
\end{center}
and
\begin{eqnarray}
 (\mathcal{L}_{\omega_1}h')\omega_2 &=& \mathcal{L}_{\omega_1}(h'\omega_2)-h'(\mathcal{L}_{\omega_1}{\omega_2}) \nonumber \\
                    &=& \alpha(\mathcal{L}_{\omega_1}{\omega_2})-h'(\mathcal{L}_{\omega_1}{\omega_2}) \nonumber \\
                    &=& \alpha(\nabla_{\omega_1}{\omega_2}-\nabla_{\omega_2}{\omega_1})-h'(\nabla_{\omega_1}{\omega_2}-\nabla_{\omega_2}{\omega_1}) \nonumber \\
                    &=& \alpha(\alpha+1)g(\omega_1,\omega_2)\xi-\alpha(\alpha+1)g(\omega_1,\omega_2)\xi \nonumber \\
                    &=& 0. \nonumber
\end{eqnarray}
 Similarly, we can demonstrate that the above results hold true if $\omega_1$ and $\omega_2$ belong to $[-\alpha]'$. For more details see \cite{dile}. Now from \eqref{4.14}, we get
\begin{equation}\label{4.15}
\{\lambda+\rho(r+4n^2)+\frac{1}{2}(D \lambda+\rho Dr)\}(\kappa+2)g(\phi \omega_1,\phi \omega_2)=0
\end{equation}
for all vector fields $\omega_1$ and $\omega_2$ on $M$. As $g(\phi \omega_1,\phi \omega_2)\neq0$ then from the for going equation we have either $\lambda=-\rho(r+4n^2)-\frac{1}{2}(D \lambda+\rho Dr)$ or $\kappa=-2$.\par
Now, for $\lambda\neq-\rho(r+4n^2)-\frac{1}{2}(D \lambda+\rho Dr)$ from the equation \eqref{4.15} we infer that $\kappa=-2$, \eqref{4.1} implies
\begin{equation}\label{4.16}
\mathcal{S}^*(\omega_1,\omega_2)=-4\{g(\omega_1,\omega_2)-\eta(\omega_1)\eta(\omega_2)\}.
\end{equation} Thus the $*$-Ricci tensor is $\eta$-Einstein manifold.\par
 \noindent Again from \eqref{4.16} and the proposition $4.1$ of \cite{dile}. We finally conclude that $M$ is locally isometric to $\mathbb{H}^{n+1}(-4)\times\mathbb{R}^n$, where $\mathbb{H}^{n+1}(-4)$ is the hyperbolic space of constant curvature $-4$.\par
\noindent Also, for any vector fields $\omega_1$ and $\omega_2$ on $M$. By hypothesis $\lambda\neq-\rho(r+4n^2)-\frac{1}{2}(D \lambda+\rho Dr)$, from the equation \eqref{4.15} we infer that $\kappa=2\alpha$. Again from $\alpha^2=-(\kappa+1)$ we get $\alpha=-1$ and $\kappa=-2$, putting the value of $\kappa$ in \eqref{4.1} $\mathcal{S}^*(\omega_1,\omega_2)=0$ i.e. $*$-Ricci flat. The proof is now complete.
\end{proof}

\section{Gradient almost $*-$R-B-S}
In this section, our focus is on the study of gradient almost $*-$R-B-Ss on Kenmotsu manifolds. We are aware that an almost $*-$R-B-S, satisfying equation \eqref{1.2} with some smooth functions $\lambda$, serves as a generalization of a Ricci-Yamabe soliton. In a previous work by Ghosh \cite{ghosh}, Ricci almost solitons on Kenmotsu manifolds were investigated, and it was proven that if a Kenmotsu metric is a gradient Ricci almost soliton and the Reeb vector field $\xi$ leaves the scalar curvature $r$ invariant, then the manifold is an Einstein manifold. To extend and generalize the aforementioned results, we consider the case of gradient almost $*-$R-B-Ss on Kenmotsu manifolds and establish the following theorem.
\begin{thm}
Suppose a Kenmotsu manifold $M^{2n+1}(\phi,\xi,\eta,g)$ admits a gradient almost $*-$R-B-S, and the Reeb vector field $\xi$ preserves the scalar curvature $r$. In that case, $(M,g)$ is an Einstein manifold with a constant scalar curvature given by $r=-n(2n-1)$.
\end{thm}
\begin{proof}
The gradient of the soliton equation \eqref{1.3} can be written for any $\omega_1$ belongs to $\chi(M)$ as
\begin{equation}\label{5.1}
\nabla_{\omega_1}Df+Q\omega_1+\{(2n-1)-(\lambda+\rho r^*)\}\omega_1+\eta(\omega_1)\xi=0.
\end{equation}
Then applying the expression of Riemannian curvature tensor $R(\omega_1,\omega_2)Df=\nabla_{\omega_1}\nabla_{\omega_2}Df-\nabla_{\omega_2}\nabla_{\omega_1}Df-\nabla_{[\omega_1,\omega_2]}Df$, we obtain
\begin{eqnarray}\label{5.2}
&& R(\omega_1,\omega_2)Df=(\nabla_{\omega_2}Q)\omega_1-(\nabla_{\omega_1}Q)\omega_2-\omega_2(\lambda+\rho r^*)\omega_1 \nonumber \\
&&+\omega_1(\lambda+\rho r^*)\omega_2  +\{\eta(\omega_2)\omega_1-\eta(\omega_1)\omega_2\}
\end{eqnarray}
for all $\omega_1$, $\omega_2$ $\in$ $\chi(M)$. Now putting $\omega_2=\xi$ in \eqref{5.2} and using \eqref{3.1} and \eqref{3.2}, we get
\begin{equation}\label{5.3}
R(\omega_1,\xi)Df=-Q\omega_1-2n\omega_1-\xi(\lambda+\rho r^*)\omega_1+\omega_1(\lambda+\rho r^*)\xi+(\omega_1-\eta(\omega_1)\xi),
\end{equation}
for any $\omega_1$ $\in$ $\chi(M)$. By virtue of \eqref{2.11}, equation \eqref{5.3} reduces to
\begin{equation} \label{5.4}
  \omega_1(f-(\lambda+\rho r^*))\xi=-Q\omega_1+\{\xi(f-(\lambda+\rho r^*))-2n+1\}\omega_1-\eta(\omega_1)\xi
\end{equation}
for any $\omega_1$ $\in$ $\chi(M)$. Now, taking an inner product of \eqref{5.4} with $\xi$ and using \eqref{2.11}, we get $\omega_1(f-(\lambda+\rho r^*))=\xi(f-(\lambda+\rho r^*))\eta(\omega_1)$. Putting this into \eqref{5.4}, we obtain
\begin{equation}\label{5.5}
Q\omega_1=\{\xi(f-(\lambda+\rho r^*))-2n+1\}\omega_1+\{\xi(f-(\lambda+\rho r^*))-1\}\eta(\omega_1)\xi
\end{equation}
for any $\omega_1$ $\in$ $\chi(M)$. This demonstrates that the manifold $(M,g)$ is an $\eta$-Einstein manifold. To further analyze the situation, we contract equation \eqref{5.2} over $\omega_1$ with respect to an orthonormal basis $\{e_i\}$, where $1\leq i\leq 2n+1$. By performing this computation, we obtain
\begin{equation}\label{5.6}
S(\omega_2,Df)=-\sum_{i=1}^{2n+1}g((\nabla_{e_i}Q)\omega_2,e_i)+\omega_2(r)-2n\omega_2(\lambda+\rho r^*)+2n\eta(\omega_2).
\end{equation}
Now, using the formula for the Riemannian manifold which is well known:
\begin{center}
$trac_g\{\omega_1\rightarrow(\nabla_{\omega_1}Q)\omega_2\}=\frac{1}{2}\omega_2(r)$
\end{center}
and \[\sum_{i=1}^{2n+1}g((\nabla_{e_i}Q)\omega_2,e_i)=1,\]
 then from \eqref{5.6}, we get
\begin{equation}\label{5.7}
S(\omega_2,Df)=\frac{1}{2}\omega_2(r)-2n\omega_2(\lambda+\rho r^*)+2n\eta(\omega_2)
\end{equation}
for any $\omega_1$ $\in$ $\chi(M)$. From \eqref{2.11}, we can compute $S(\xi,Df)=-2n\xi(f)$, putting this in \eqref{5.7}, we get $\xi(r)=4n\{\xi(\lambda+\rho r^*-f)-1\}$. Using this in the trace of \eqref{3.2}, we get $\xi(f-(\lambda+\rho r^*))=\frac{r}{n}+2n$. By this result, equation \eqref{5.5} reduces to
\begin{equation}\label{5.8}
Q\omega_1=(\frac{r}{n}+1)\omega_1+\{\frac{r}{n}+2n-1\}\eta(\omega_1)\xi
\end{equation}
for any $\omega_1$ $\in$ $\chi(M)$. Based on our assumption that $\xi(r)=0$, we can take the trace of equation \eqref{3.2} and utilize equation \eqref{5.8}. This leads us to the conclusion that $r=-n(2n-1)$. Consequently, the desired result follows from equation \eqref{5.8}.
\end{proof}
\section{Example of a 5-dimensional almost Kenmotsu manifold admitting an almost $*-$R-B-S}
Let $M=\{(x_1,y_1,z_1,u_1,v_1)(\neq 0) \in \mathbb{R}^5\}$ be a 5-dimensional manifold, where $(x_1,y_1,z_1,u_1,v_1)$ be the standard coordinates in $\mathbb{R}^5$. Now, let us consider a orthonormal basis $\{e_1, e_2, e_3, e_4, e_5\}$ of vector fields on $M$, where,
\begin{align*}
  e_1&=v_1 \frac{\partial}{\partial x_1}, & e_2&=v_1\frac{\partial}{\partial y_1}, & e_3&=v_1\frac{\partial}{\partial z_1}, & e_4&=v_1\frac{\partial}{\partial u_1}, & e_5&=-v_1\frac{\partial}{\partial v_1}.
\end{align*}
Define $(1,1)$ tensor field $\phi$ as follows:
\begin{align*}
   \phi(e_1)&=e_2, & \phi(e_2)&=-e_1, & \phi(e_3)&=e_4, & \phi(e_4)&=-e_3, & \phi(e_5)&=0.
\end{align*}

\noindent The Riemannian metric is given by
\[(g_{ij})=\left (\begin{array}{ccccc}
1 & 0 & 0 & 0 & 0\\
0 & 1 & 0 & 0 & 0\\
0 & 0 & 1 & 0 & 0\\
0 & 0 & 0 & 1 & 0\\
0 & 0 & 0 & 0 & 1\\
\end{array}\right )\]
and $\eta(\omega_1)=g(\omega_1,e_5)$ for any $\omega_1\in\chi(M^5)$. Then, for $\xi=e_5$, the relations \eqref{2.1}, \eqref{2.2} and \eqref{2.3} are satisfied. So, $M(\phi,\xi,\eta,g)$ is an almost contact metric manifold.\par
\noindent The non-zero components of the Levi-Civita connection $\nabla$ (using Koszul’s formula) are
\begin{eqnarray}\label{3.p}
\begin{array}{ccccc}
\nabla_{e_1}{e_1}=\nabla_{e_2}{e_2}=\nabla_{e_3}{e_3}=\nabla_{e_4}{e_4}=-e_5, \\
 \nabla_{e_1}{e_5}=e_1,~  \nabla_{e_2}{e_5}=e_2,~  \nabla_{e_3}{e_5}=e_3,~ \nabla_{e_4}{e_5}=e_4.
\end{array}
\end{eqnarray}
By virtue of this we can verify \eqref{2.8} and therefore $M^5(\varphi,\xi,\eta,g)$ is a Kenmotsu manifold.\par
\noindent Using the expression of curvature tensor $R(\omega_1,\omega_2) = [\nabla_{\omega_1},\nabla_{\omega_2}] - \nabla_{[\omega_1,\omega_2]}$, we now compute the following non-zero components
\begin{align*}
R(e_1,e_2)e_2&=-e_1, & R(e_1,e_3)e_3&=-e_1, & R(e_1,e_4)e_4&=-e_1,\\
R(e_1,e_5)e_5&=-e_1, & R(e_1,e_2)e_1&=e_2, & R(e_1,e_3)e_1&=e_3,\\
R(e_1,e_4)e_1&=e_4, & R(e_1,e_5)e_1&=e_5, & R(e_2,e_3)e_2&=e_3,\\
R(e_2,e_4)e_2&=e_4, & R(e_2,e_5)e_2&=e_5, & R(e_2,e_3)e_3&=-e_2,\\
R(e_2,e_4)e_4&=-e_2, & R(e_2,e_5)e_5&=-e_2, & R(e_3,e_4)e_3&=e_4,\\
R(e_3,e_5)e_3&=e_5, & R(e_3,e_4)e_4&=-e_3, & R(e_4,e_5)e_4&=e_5,\\
R(e_5,e_3)e_5&=e_3, & R(e_5,e_4)e_5&=e_4.
\end{align*}
Using this, we compute the components of the Ricci tensor as follows,\\
\centerline{$S(e_i,e_i)=-4$ ~for $i=1, 2, 3, 4, 5$}, and therefore
\begin{eqnarray}\label{3.p1}
S(\omega_1,\omega_2)=-4 g(\omega_1,\omega_2) ~~for ~all~ \omega_1, \omega_2\in \chi(M^5).
\end{eqnarray}
Also from \eqref{3.3} and \eqref{3.p1}, we get
\begin{eqnarray}\label{6.p}
S^*(\omega_1,\omega_2)=-4[g(\omega_1,\omega_2)-\eta(\omega_1)\eta(\omega_2)]  ~~for ~all~ \omega_1, \omega_2\in \chi(M^5).
\end{eqnarray}
Let $f:M\rightarrow \mathbb{R}$ be a smooth function defined by
\begin{eqnarray}\label{3.p11}
f(x_1, y_1, z_1, u_1, v_1)={x_1}^2+{y_1}^2+{z_1}^2+{u_1}^2+\frac{{v_1}^2}{2}.
\end{eqnarray}
The potential vector field is given by,
\begin{eqnarray}\label{3.p111}
\omega=Df=2 x_1 \frac{\partial}{\partial x_1}+2y_1 \frac{\partial}{\partial y_1}+2z_1 \frac{\partial}{\partial z_1}+2u_1\frac{\partial}{\partial u_1}+v_1\frac{\partial}{\partial v_1}.
\end{eqnarray}
Then with the help of \eqref{3.p} we can show that
\begin{eqnarray}\label{3.p2}
(\mathcal{L}_\omega g)(\omega_1,\omega_2)=2\{ g(\omega_1,\omega_2)-4\eta(\omega_1)\eta(\omega_2)\},
\end{eqnarray}
for all $\omega_1$, $\omega_2$ $\in\chi(M^5)$.
So, combining \eqref{6.p} and \eqref{3.p2}, we observe that soliton Eq. \eqref{1.2} holds for $\lambda=16\rho-3$ i.e., the metric $g$ is gradient almost $*-$R-B-S with this potential vector field $V=Df$ and $\lambda=16\rho-3$. Also, equation \eqref{6.p} shows that the manifold becomes $\eta$-Einstein. So, it satisfies \textbf{Theorem 3.5}.
\section{Conclusion and Physical Applications}
  The notion of an almost $*-$R-B-S is a recent and significant concept in the field of differentiable manifolds, with applications in mathematical physics, quantum cosmology, quantum gravity, and black holes. It provides a framework for understanding geometric and physical phenomena in spacetimes with relativistic viscous fluid, heat flux, stress, dark and dust fluids, and radiation era. The concept of an $*$-almost R-B-S and gradient almost $*-$R-B-S can be examined to discuss the RG flow of mass in 2 dimensions. An almost R-B-S is crucial in understanding energy and entropy concepts in general relativity, analogous to the heat equation that governs the heat loss of an isolated system towards thermal equilibrium.

\noindent In our study, we explore the kinetic and potential nature of relativistic spacetime and apply it to cosmology and general relativity. We present physical models of three classes, namely, shrinking, steady, and expanding perfect and dust fluid solutions of an almost $*-$R-B-S spacetime. The shrinking case ($\Omega < 0$) exists on a minimal time interval $-1 < t < b$, where $b < 1$. The steady case ($\Omega = 0$) exists for all time, while the expanding case ($\Omega > 0$) exists on a maximal time interval $a < t < 1$, where $a > -1$. These three classes illustrate examples of ancient, eternal, and immortal solutions, respectively. The physical applications of almost $*-$R-B-Ss can be further explored, as briefly discussed in references \cite{Duggal, Woolgar}.

\noindent Our article raises several questions that can be addressed in future research:

(i) Is \textbf{Theorem 5.1} applicable if the scalar curvature $r$ is not invariant?

(ii) Which results from our paper hold true for nearly Kenmotsu manifolds, $f$-Kenmotsu manifolds, or K\"{a}hler manifolds?
\section*{Conflict of interest}
The corresponding author states that there is no conflict of interest.

\section*{Data Availability}
No data is required in this manuscript.

\section*{Acknowledgments} \noindent This work was supported and funded by the Deanship of Scientific Research at Imam Mohammad Ibn Saud Islamic University (IMSIU) (grant number IMSIU-RP23105). The seventh author (Richard Pincak) is thankful to Slovak Science Agency for providing partial financial by VEGA fund under grant number VEGA 2/0076/23.\\

\end{document}